\documentclass[11pt,reqno]{amsart}
\usepackage[english]{babel}
\usepackage{bm}
\usepackage{bbold}
\usepackage{xcolor}
\usepackage{graphicx}
\usepackage{amssymb,amsmath}

\usepackage[margin=1in]{geometry}
\usepackage[pagebackref, colorlinks = true, linkcolor = blue, urlcolor  = blue, citecolor = red]{hyperref}

\newcommand{\otimesmin}{\otimes_{\mathrm{min}}}
\newcommand{\otimesmax}{\otimes_{\mathrm{max}}}
\newcommand{\kk}{\mathbf k}
\newcommand{\<}{\langle}

\renewcommand{\>}{\rangle}

\newcommand{\Tr}{\mathrm{Tr\,}}
\newcommand{\ou}{\mathbb{1}}

\newtheorem{thm}{Theorem}[section]
\newtheorem*{thm*}{Theorem}
\newtheorem{lemma}[thm]{Lemma}
\newtheorem{cor}[thm]{Corollary}

\newtheorem{prop}[thm]{Proposition}
\theoremstyle{remark}
\newtheorem{ex}[thm]{Example}
\newtheorem{remark}[thm]{Remark}

\title{Assemblages and steering in general probabilistic theories}

\author{Anna Jen\v {c}ov\'a}
\email{jenca@mat.savba.sk}
\address{Mathematical Institute, Slovak Academy of Sciences, Bratislava, Slovakia}

\begin{document}

\maketitle

\begin{abstract} We study steering in the framework of general probabilistic theories. 
We show that for dichotomic assemblages, steering can be characterized in terms of  a  certain tensor
cross norm, which is also related to a steering degree given by
steering robustness. Another contribution is the observation that steering in GPTs can be
conveniently treated using Choquet theory for probability measures on the state space. 
In particular, we find a variational expression  for  universal steering degree for
dichotomic assemblages and obtain conditions characterizing unsteerable states analogous
to some conditions recently found for the quantum case. The setting also enables us to
rather easily extend the results to infinite dimensions and arbitrary numbers of
measurements with arbitrary outcomes.

\end{abstract}

\section{Introduction}

 EPR steering was first described by Schr\"odinger in 1936
 \cite{schrodinger1936probability} as a  bipartite scenario where  one party can steer the
 state of a distant party by performing local measurements,  in a way that cannot be
 explained by classical correlations.  A precise definition and a systematic treatment was
 given in   \cite{wiseman2007steering, jones2007entanglement}, where quantum steering was interpreted in operational terms, as the
possibility of certifying entanglement when  one party is untrusted. 
In this way, steering became an important resource in quantum information theory and  has attracted a lot of attention due to
applications as well as for its relations to other
nonclassical phenomena, such as entanglement, Bell nonlocality or incompatibility of
measurements. In recent years, methods for  characterizing and quantifying steering were developed in the literature
that can be efficiently evaluated by SDPs, see \cite{cavalcanti2016quantum} for an
overview. For a recent
review  of quantum steering see \cite{uola2020quantum}.

Quantum steering can be described as follows. Assume Alice and Bob share a bipartite state
$\rho_{AB}$. After Alice performs a measurement $M_x$  on her side  and
obtains the result $a$, Bob is left with the conditional (nonnormalized) state
$\rho_{a|x}=(M_{a|x}\otimes id)(\rho_{AB})$. 
The assemblage $\{\rho_{a|x}\}$ of conditional states admits a local hidden state (LHS) model if there is an
ensemble of states $\{\lambda(\omega),\rho_\omega\}$ from which $\rho_{a|x}$ are obtained by a set of conditional
probabilities $\{p(a|x,\omega)\}$, in this case Bob is not convinced that the state
$\rho_{AB}$ was entangled. If no LHS model exists, Bob can be sure of both entanglement
and incompatibility of Alice's measurements. However, there are entangled states that
are unsteerable, which means that the assemblages  obtained by any measurement always
have a LHS model, this was already observed in \cite{wiseman2007steering, jones2007entanglement}. 
Some steerability criteria were recently obtained in
\cite{nguyen2018quantum,nguyen2019geometry,nguyen2020some} through a geometric approach.

Nonclassical  phenomena are often studied in a broader framework of general probabilistic
theories (GPTs) \cite{barrett2007information}, which was used for better understanding the
operational features of Bell nonlocality, incompatibility of measurements and steering,
and their relations
\cite{barnum2013ensemble,busch2014steering,banik2015measurement,plavala2016allmeasurements, plavala2017conditions,jencova2018incompatible}.
GPTs  include the classical and quantum theory, as well as the PR boxes exhibiting maximal
violations of Bell inequalities  \cite{popescu1994quantum}, quantum channels
\cite{jencova2017conditions} or post quantum
steering \cite{cavalcanti2021witworld}. Another
motivation comes from quantum foundations, where the aim is to characterize quantum theory
among physical theories. 

In the GPT setting, the relations of these phenomena to some  mathematical concepts can be
revealed. For example,  nonclassical correlations can be expressed in
terms of tensor cross norms in Banach spaces \cite{aubrun2018universal}. In \cite{bluhm2020incompatibility},
incompatibility of measurements in GPTs is characterized and quantified using different
mathematical points of view: extendability of maps, theory of generalized spectrahedra and
tensor cross norms. 

The aim of the present work is to study steering in GPTs. We first restrict to dichotomic
steering, where Alice uses only dichotomic measurements. In this case, we show that
steering is characterized by tensor cross norms in very much the same way as obtained in
\cite{bluhm2020incompatibility}. In fact, the connection is immediate from the formal
relation of the two notions of steering and incompatibility: in the framework of GPTs they
can be seen as equivalent. The results then follow straightforwardly from those in
\cite{bluhm2020incompatibility}, but we give different proofs, more clearly related to the
structure of tensor products with hypercubic state spaces. 

In the second part, we show that steering can be represented using Choquet order on the
set of Radon probability measures over the state space.
More precisely, the assemblages are  represented by sets of (simple)
probability measures and a LHS model is  given by another probability
measure that is an upper bound for the assemblage in Choquet order. We then obtain a
convenient description of LHS models and steering using classical results in Choquet
theory, which naturally includes infinite assemblages.
In particular, we show that we can always use an LHS model concentrated on the extreme
boundary of the state space, equivalently, the LHS model is given by a boundary
measure.  Further, we may assume that this measure is invariant under the group of
transformations leaving the assemblage invariant.  Using this representation, we
find a variational expression for the steering degree for dichotomic
assemblages that is independent of the size of the assemblage. In the quantum case, 
this expression coincides  with the  quantum steering constants  obtained in \cite{bluhm2021maximal}
and relates to the 1-summing constant for centrally symmetric state spaces.
We also find  characterizations of (one way)  unsteerable states in a GPT,
 which are similar to those obtained in \cite{nguyen2018quantum, nguyen2019geometry} in
 the quantum case.

\section{Notations and preliminaries}

In this paper we follow the general assumptions and formalism of GPTs used in
\cite{lami2018non}. See e.g. \cite{plavala2021general}
for another recent exposition with a convenient diagrammatic presentation.

A system of a general probabilistic theory (GPT) is a triple $(V,V^+,\ou)$,
where $V$ is a real vector space of finite dimension, $V^+$ is a closed convex cone which is pointed 
($V^+\cap(-V^+)=\{0\}$) and separating $(V=V^+-V^+)$ and $\ou$ is a distinguished element of
the dual vector space $V^*=A$. In $A$, the dual cone $A^+$ is defined as
\[
A^+=\{f\in A,\ \<f,v\>\ge 0,\ \forall v\in V^+\},
\]
here $\<\cdot,\cdot\>$ denotes the duality $A\times V\to \mathbb R$. We will use the
notation $\le$ for the ordering induced by $V^+$ or $A^+$ in their respective spaces. In
this  ordering, $\ou$ is assumed to be an order unit, which means that for all $f\in A$
there is some $\alpha>0$ such that $\alpha\ou\pm f\in A^+$, or, equivalently, that $\ou$ is an
interior element in $A^+$. The order unit norm in $A$ is defined as
\[
\|f\|_\ou=\inf\{\lambda>0,\ \lambda\ou\pm f\in A^+\}.
\]

The subset
\[
K:=\{\rho\in V^+,\ \<\ou,\rho\>=1\}.
\]
is interpreted as the set of states and called the state space. Clearly, $K$
 is a compact convex subset of $V$ and a base of the cone  $V^+$. The dual of
 $\|\cdot\|_\ou$ in $V$ is the base norm 
\[
\|v\|_V=\inf\{\<\ou,v_++v_-\>,\ v=v_+-v_-,\ v_+,v_-\in V^+\} = \sup_{\pm f\le \ou}
\<f,v\>.
\]

Elements of the unit interval $E=\{f\in A,\ 0\le f\le \ou\}$ are called effects.
The effects can be identified with affine maps $K\to [0,1]$ and are interpreted as
dichotomic (or yes-no) measurements of the system: the value $\<f,\rho\>\in [0,1]$ gives the
probability that the measurement represented by the effect $f$ gives outcome ''yes'' if
the system is in the state $\rho\in K$. Similarly, measurements with $n$ outcomes are
represented by collections $f_1,\dots,f_n\in E$, $\sum_i f_i=\ou$, where $\<f_i,\rho\>$ is
interpreted as the probability of i-th outcome in the state $\rho$.

Given two systems $(V_A,V_A^+,\ou_A)$ and $(V_B,V_B^+,\ou_B)$, the composite system is
$(V_{AB}, V_{AB}^+,\ou_{AB})$, where $V_{AB}=V_A\otimes V_B$, $\ou_{AB}=\ou_A\otimes
\ou_B$ and $V_{AB}^+$ is a positive cone in $V_{AB}$ satisfying
\[
V_A^+\otimesmin V_B^+\subseteq  V_{AB}^+\subseteq  V_A^+\otimesmax V_B^+.
\]
Here $V_A^+\otimesmin V_B^+$ is the minimal cone
\[
V_A^+\otimesmin V_B^+=\{\sum_i v_{i,A}\otimes v_{i,B},\ v_{i,A}\in V^+_A, v_{i,B}\in V^+_B\}
\]
containing all separable states and $V_A^+\otimesmax V_B^+$ is the maximal cone
\[
V_A^+\otimesmax V_B^+=\{v_{AB}\in V_{AB}, \ \<f_A\otimes f_B,v_{AB}\>\ge 0,\ \forall
f_A\in E_A,\ f_B\in E_B\}
\]
for which all separately prepared measurements are valid. 
Similarly, the state space
$K_{AB}$ satisfies
\[
K_A\otimesmin K_B\subseteq K_{AB} \subseteq K_A\otimesmax K_B.
\]
As it was recently proved in \cite{aubrun2021entangleability}, the inclusion $K_A\otimesmin K_B\subseteq  K_A\otimesmax K_B$ is strict,
unless $K_A$ or $K_B$ is a simplex. The states in $K_A\otimesmax K_B$ are called
separable, all other states in $K_A\otimesmin K_B$ are entangled.

Some of the  basic examples are described below.

\begin{ex} \textbf{Classical theory.}  In the classical GPT, the systems have the form
$(\mathbb R^n, \mathbb R^n_+, (1,\dots,1))$, where $\mathbb R_+^n$ is the simplicial cone
generated by the positive half-axes. The state space is the simplex
$\Delta_n=\{(x_1,\dots,x_n),\ x_i\ge 0,\ \sum_i x_i=1\}$.  For any state space $K$ we
have
$\Delta_n\otimesmax K=\Delta_n\otimesmin K=K^{\oplus n}$, the convex direct sum of
$n$-copies of $K$. 
\end{ex}

\begin{ex}
\textbf{Quantum theory.} Here $V=M_n^{sa}$ is the space of $n$ by $n$ complex
hermitian matrices,  $V^+=M_n^+$ is the cone of positive semidefinite matrices and
$\ou=I$, the identity matrix. The state space is the set of density matrices
$D_n=\{\rho\in M_n^+,\ \Tr \rho=1\}$. The tensor product of the cones $M_m^+$ and $M_n^+$ is the
cone $M_{mn}^+$ of positive definite matrices in $M_{mn}^{sa}=M_m^{sa}\otimes M_n^{sa}$. 
\end{ex}

\begin{ex}\label{ex:GPTS}
Note that for any compact convex subset $S$ in the Euclidean space $\mathbb R^g$ we can
construct a triple $(V_S,V_S^+,\ou_S)$  with the state space isomorphic to $S$: put
$V_S=\mathbb R^{g+1}$ and
 $K=\{(1, x),\ x\in S\}$. We define $V_S^+$ as the cone generated by $K$ and $\ou_S=(1,0)$.
\end{ex}

\begin{ex}\label{ex:cs}
\textbf{Centrally symmetric state spaces.} In the previous example, assume that $S$ is
the unit ball of a norm $\|\cdot\|$ in $\mathbb R^g$. Then $K=\{(1,x),\ \|x\|\le 1\}$ and
$V_S^+=\{(t,x),\ \|x\|\le t\}$. Note that the dual space $(A_S,A_S^+)$ has the same form for the dual norm
$\|\cdot\|^*$ and the central element $(1,0)$ is an order unit in both $(V_S,V_S^+)$ and
$(A_S,A_S^+)$. In this
case we have for the base norm (in $V_S$) $\|(s,x)\|_{V_S}=\max\{|s|,\|x\|\}$ and the order unit
norm (in $A_S$) becomes $\|(t,\varphi)\|_{(1,0)}=|t|+\|\varphi\|^*$.

We now look at the tensor product with any triple $(V,V^+,\ou)$. We will use the obvious 
identifications $\mathbb R^{g+1}\otimes V\cong V^{g+1}\cong V\oplus V^g$, so the first
copy is distinguished. Let  $(y_0,y)\in V^{g+1}$, $y=(y_1,\dots,y_g)$, then it is easily
checked that $(y_0,y)\in V_S^+\otimesmin V^+$ if and only if
\begin{equation}\label{eq:cs_min}
y=\sum_j x_j\otimes z_j,\quad \|x_j\|=1,\ z_j\in V^+,\ \sum_j z_j\le y_0
\end{equation}
and $(y_0,y)\in V_S\otimesmax V^+$ if and only if
\begin{equation}\label{eq:cs_max}
\sum_i \varphi_i y_i\le y_0,\qquad \forall \varphi\in \mathbb R^g,\ \|\varphi\|^*=1.
\end{equation}
Note that using this condition for $\pm \varphi$, we obtain that  $y_0\in V^+$ and $y_0=0$
only if $y=0$.

An important centrally symmetric example is the qubit system $(M_2^{sa},M_2^+,I_2)$, where the state space is
isomorphic to the unit ball in $\ell^3_2$.  We will also frequently use the hypercubic
systems obtained from $\ell_\infty^g$, we will  denote the corresponding triple as
$(V_g,V_g^+,\ou_g)$.  The state space
\[
S_g:=\{(1,z_1,\dots,z_g),\ |z_i|\le 1\}
\]
is isomorphic to the hypercube $[-1,1]^g$.

\end{ex}

We now consider the norms obtained in $V_{AB}$ from the tensor product of the Banach
spaces  $V_A$ and  $V_B$ equipped with their respective base norms. A norm $\|\cdot\|$ on $V_{AB}$ is
called a (reasonable) cross norm if both $\|\cdot\|$ and its dual norm are multiplicative
on simple tensors. Equivalently,
\[
\|\cdot\|_{\epsilon(A,B)}\le \|\cdot\|\le \|\cdot\|_{\pi(A,B)},
\]
where $\epsilon(A,B)$ denotes the injective cross norm given by
\[
\|v_{AB}\|_{\epsilon(A,B)}=\sup\{ \<f_A\otimes f_B,v_{AB}\>,\ \ou_A\pm f_A\ge 0,\ou_B\pm
f_B\ge 0\}
\]
and $\pi(A,B)$ denotes the projective cross norm
\[
\|v_{AB}\|_{\pi(A,B)}=\inf\{ \sum_i \|v_{i,A}\|_{V_A}\|v_{i,B}\|_{V_B},\ v_{AB}=\sum_i
v_{i,A}\otimes v_{i,B}\}.
\]
It was proved in \cite{aubrun2018universal} that the base norm for the composite system
$(V_{AB},V^+_{AB},\ou_{AB})$ is a reasonable cross norm, and that it coincides with the
projective cross norm $\pi(A,B)$ for the separable cone $V_{AB}^+=V_A^+\otimesmin V_B^+$. 
For completeness,  we give a proof for the following equivalent formulation of the latter
statement.

\begin{thm}\label{thm:separable}  A bipartite state $\rho_{AB}\in K_A\otimesmax K_B$ is separable
if and only if $\|\rho_{AB}\|_{\pi(A,B)}\le 1$. 
\end{thm}

\begin{proof} Assume $\rho_{AB}$ is separable, then $\rho_{AB}=\sum_i\lambda_i
\rho_{i,A}\otimes \rho_{i,B}$ for $\rho_{i,A}\in K_A$, $\rho_{i,B}\in K_B$ and probabilities
$\lambda_i$. By definition of the projective cross norm,
\[
\|\rho_{AB}\|_{\pi(A,B)}\le \sum_i\lambda_i\|\rho_{i,A}\|_{V_A}\|\rho_{i,B}\|_{V_B}=1.
\]
Assume the converse and let $y_{i,A}\in V_A$ and $y_{i,B}\in V_B$ be such that
$\rho_{AB}=\sum_i y_{i,A}\otimes y_{i,B}$ and $\sum_i \|y_{i,A}\|_{V_A}\|y_{i,B}\|_{V_B}\le
1$. Then $y_{i,A}=y_{i,A}^+-y_{i,A}^-$ with $y_{i,A}^\pm\in V_A^+$ and
$\|y_{i,A}\|_{V_A}=\|y_{i,A}^+\|_{V_A}+\|y_{i,A}^-\|_{V_A}$, similarly for $y_{i,B}$.
It follows that $\rho_{AB}=\rho_{AB}^+-\rho_{AB}^-$, with
\[
\rho_{AB}^+=\sum_i(y_{i,A}^+\otimes y_{i,B}^++y_{i,A}^-\otimes y_{i,B}^-)\in V_A^+\otimesmin
V_B^+
\]
and 
\[
\rho_{AB}^-=\sum_i(y_{i,A}^+\otimes y_{i,B}^-+y_{i,A}^-\otimes y_{i,A}^+)\in V_A^+\otimesmin
V_B^+
\]
With $\ou_{AB}=\ou_A\otimes \ou_B$, we obtain
\[
1=\<\ou_{AB},\rho_{AB}\>=\<\ou_{AB},\rho_{AB}^+-\rho_{AB}^-\>\le
\<\ou_{AB},\rho_{AB}^++\rho_{AB}^-\>=\sum_i \|y_{i,A}\|_{V_A}\|y_{i,B}\|_{V_B}\le 1,
\]
whence $\rho_{AB}^-=0$ and $\rho_{AB}=\rho_{AB}^+$ is separable.

\end{proof}

%
%

\section{Steering and tensor norms}\label{sec:tensor}

\subsection{Conditional states and assemblages}

Let $\sigma_{AB}\in K_{AB}$ be a state of the composite system $(V_{AB}, V_{AB}^+,\ou_{AB})$
and let $\{f_x\}_{x=1}^g$ be a (finite) collection of measurements on the system
$(V_A,V_A^+,\ou_A)$, with effects $f_{a|x}$ and outcomes  $a\in \Omega_x$. Viewing any $f\in E$ as an affine
function over $K$, we may define the conditional states
\[
\rho_{a|x}:= (f_{a|x}\otimes \mathrm{id}_B)(\sigma_{AB})\in V_B^+.
\]
For all $x=1,\dots, g$, we have $\sum_{a\in \Omega_x} \rho_{a|x}=\sigma_B=(\ou_A\otimes
\mathrm{id}_B)(\sigma_{AB})$. More generally, any collection
\[
\{\rho_{a|x}\in V_B^+,\ \sum_a \rho_{a|x}=\sigma_B,\ \forall x=1,\dots,g\}
\]
is called an assemblage with barycenter $\sigma_B$. The tuple $\kk=(k_1,\dots,k_g)$ with 
$k_x=|\Omega_x|$, $x=1,\dots,g$ determines the shape of the assemblage. We will say
that the assemblage is dichotomic if $k_x=2$ for all $x$, in this case we will use
$\Omega_x=\{+,-\}$
as the set of labels.

The next result  shows that 
if we do not restrict the choice of the system $V_A$ and allow maximal tensor  products, all
assemblages can be obtained as conditional states for  a bipartite state   and some
collection of measurements on $V_A$ of a corresponding shape. 

\begin{thm} \label{thm:assemblages}
Let $\kk=(k_1,\dots,k_g)$ and let $S$ be the Cartesian product $S=S_{\kk}:=\Pi_{x=1}^g
\Delta_{k_x}$.  For any  system $(V,V^+,\ou)$ with state space $K$, the set of all assemblages of shape
$\kk$ can be identified with the tensor product $S\otimesmax K$. In particular, 
there is a system $(V_S,V_S^+,\ou_S)$ with state space $S$ and with a canonical set of
measurements $\{p_x\}$ (identified with projections $p_x:S\to \Delta_{k_x}$ onto the $x$-th
component),  such that for any assemblage $\{\rho_{a|x}\}$ in $V^+$ of shape $\kk$,  
there is a (unique) state $\xi\in S\otimesmax K$ such
that 
\[
\rho_{a|x}= (p_{a|x}\otimes \mathrm{id}_K)(\xi),\qquad \forall a,x.
\]
\end{thm}

\begin{proof} We will give a proof for the case of dichotomic assemblages, mostly to
introduce some notations needed later. For the proof in general see
\cite{jencova2018incompatible}. 

So let $\kk=(2,\dots,2)$. By the isomorphism $(\lambda,1-\lambda)\mapsto 2\lambda-1$,
we will identify   $S=\Delta_2^g\simeq [-1,1]^g$. Let $(V_g,V_g^+,\ou_g)$ be the system with centrally
symmetric state space $S_g\simeq [-1,1]^g$, see Example \ref{ex:cs}.  It is easily seen from
\eqref{eq:cs_max} that the maximal tensor product 
$S_g\otimesmax K$ can be identified with the subset in $V^{g+1}$ of elements of the form $(\sigma, y_1,\dots,y_g)$ with
$\sigma\in K$ and $\sigma\pm y_x\in V^+$. Clearly,  $\{\frac12(\sigma\pm y_x)\}$ is a
dichotomic assemblage with barycenter $\sigma$. Conversely, for any assemblage $\{\rho_{\pm|x}\}$
with barycenter $\sigma$, $(\sigma, y_1,\dots, y_g)$ with
$y_x:=\rho_{+,x}-\rho_{-|x}$ is an element of $S_g\otimesmax K$ and it is clear that this
established a one-to-one correspondence. The effects $p_{\pm|x}$ have the form
$\frac12(1, \pm e_x)$, where $\{e_x\}_{x=1}^g$ is the standard basis in $\mathbb R^g$.

\end{proof}

It is well known that any assemblage in quantum theory can be obtained with $V_A\simeq
V_B$. In the general case, this depends on the properties of the system $V_B$ as well as
the choice of the tensor product. For more information see \cite{barnum2013ensemble}.

\subsection{Local hidden state models and steering}\label{sec:lhs}

As in the quantum case, we say that an assemblage of conditional states
$\{\rho_{a|x}=(f_{a|x}\otimes id_B)(\sigma_{AB})\}$ admits a local hidden state (LHS) model if there is some (finite) set
$\Lambda$, a probability measure
$q\in \mathcal P(\Lambda)$, conditional probabilities $q(x|a,\lambda)$ and elements 
 $\rho_\lambda\in K$ such that 
\begin{equation}\label{eq:unsteerable}
\rho_{a|x}=\sum_{\lambda\in \Lambda} q(\lambda)q(a|x,\lambda)\rho_\lambda,\qquad a\in
\Omega_x,\ x=1,\dots, g.
\end{equation}
We say that a bipartite state $\sigma_{AB}$ is ($A\to B$) steerable 
if there is a set $\{f_x\}$ of measurements on the system $A$ such that the assemblage 
 $\{\rho_{a|x}\}$ does not admit a LHS model.  If no such collection of measurements exists, the
 state is unsteerable. We may also restrict the set of measurements, so we say that, for
 example, a state is unsteerable by dichotomic measurements if all corresponding
 dichotomic assemblages admit a LHS.

For a general assemblage $\{\rho_{a|x}\}$ satisfying \eqref{eq:unsteerable} we will also say
that the assemblage is classical.  We will show later  that we may equivalently formulate
the condition  in \eqref{eq:unsteerable} with the  
set $\Lambda$ replaced by the  (possibly infinite) set $\partial_e K$ of extremal points of $K$
(pure states), common for all LHS models.

\begin{thm}\label{thm:lhs} Let $\{\rho_{a|x}\}$ be an assemblage of shape $\kk$. Then
the assemblage is classical  if and only if the corresponding element of
$S_{\kk}\otimesmax K$ is separable.

\end{thm}

\begin{proof} We again prove the statement for dichotomic assemblages that we focus on in
this section. For a general proof, see  \cite{jencova2018incompatible}.

We have seen in the  proof of Theorem \ref{thm:assemblages} that the element in $S_g\otimesmax K$ corresponding
 to the assemblage $\{\rho_{\pm|x}\}$ is $(\sigma,y_1,\dots,y_g)$, with
$y_x=\rho_{+|x}-\rho_{-|x}$ and $\sigma$ the barycenter. Observe that the LHS model
\eqref{eq:unsteerable} for $\{\rho_{\pm|x}\}$ is equivalent to
\begin{equation}\label{eq:cs_min1}
y_x= \sum_\lambda h_\lambda(x)\phi_\lambda,\qquad x=1,\dots,g
\end{equation}
for some $h_\lambda(x)\in [-1,1]$ and $\phi_\lambda\in V^+$,
$\sum_\lambda\phi_\lambda=\sigma$. Indeed, if \eqref{eq:unsteerable} holds then we may put
$h_\lambda(x)=q(+|x,\lambda)-q(-|a,\lambda)$ and $\phi_\lambda=q(\lambda)\rho_\lambda$.
Conversely, we obtain a LHS model from \eqref{eq:cs_min1} by setting
$q(\pm|x,\lambda)=\frac12(1\pm h_\lambda(x))$ and normalizing the elements $\phi_\lambda$.
This amounts to
\[
(\sigma,y)=\sum_\lambda h_\lambda \otimes \phi_\lambda
\]
which can be seen to be equivalent to the characterization of $S_g\otimesmin K$ in
\eqref{eq:cs_min}, see  Example
\ref{ex:cs}.

\end{proof}

Our first characterization of steering by tensor cross norms  follows immediately from Theorems \ref{thm:assemblages}, \ref{thm:lhs}
and \ref{thm:separable}. 

\begin{cor}\label{cor:lhs} Let $\{\rho_{a|x}\}$ be an assemblage and let
$K_A=S$.  Let $\xi_{AB}\in K_A\otimes K_B$ be such that $(p_{a|x}\otimes
id)(\xi_{AB})=\rho_{a|x}$ as in  Theorem \ref{thm:assemblages}. 
Then the assemblage is classical  if and only if 
$\|\xi_{AB}\|_{\pi(A,B)}\le 1$.

\end{cor}

\begin{remark} Note that if $K_A=S$ as in Corollary \ref{cor:lhs}, any $\sigma_{AB}\in
K_A\otimes K_B$ is
unsteerable if and only if  $\{(p_{a|x}\otimes id_B)(\sigma_{AB})\}$ admits
a LHS model if and only if
$\sigma_{AB}$ is separable. It is well known and immediately seen that a separable state is
always unsteerable: indeed,  we can view any set of measurements $\{f_x\}$ as a map $f_A:
K_A\to S$ such that $f_A\otimes \mathrm{id}_B$ maps $\sigma_{AB}$ to the element in
$S\otimesmax K_B$ corresponding to the assemblage of conditional states. Since the map
$f_A$ preserves the state spaces, it is a contraction with respect to the base norms and we
have from the properties of the projective cross norm that 
\[
\|(f_A\otimes \mathrm{id}_B)(\sigma_{AB})\|_\pi\le \|\sigma_{AB}\|_\pi\le 1.
\]
here $\pi$ always denotes the projective cross norm for the respective base norms. On the
other hand, there are unsteerable entangled quantum states.

\end{remark}

\subsection{Dichotomic assemblages and tensor cross norms}

Let $\sigma\in K$ be an interior element, so that 
$\sigma\in \mathrm{int}(V^+)$. We will show that classical dichotomic assemblages with barycenter $\sigma$ 
  can be characterized by  tensor cross norms in $\ell_\infty^g\otimes V$, if we choose an appropriate
norm in $V$, depending on $\sigma$. We will denote the set of all such assemblages with
$g$ elements  by
$\mathcal A_{2,\sigma}^g$. 

Note first that the barycenter  $\sigma$ is an order unit in
$(V,V^+)$, so that  we can (formally) define a system as $(A,A^+,\sigma)$, where
$\sigma$ is the unit effect and the state space becomes 
\[
K^\sigma=\{h\in A^+,\ \<h,\sigma\>=1\}.
\]
Let us denote the corresponding base norm (in $A$) by $\|\cdot\|^\sigma$ and the dual
order unit norm (in $V$) by $\|\cdot\|_\sigma$. 

We have seen in the proof of Theorem \ref{thm:assemblages} that assemblages in $\mathcal A_{2,\sigma}^g$  can be identified with
elements  $y=(y_1,\dots,y_g)\in V^g\simeq \mathbb R^g\otimes V$ such that  $\pm y_x\le \sigma$, equivalently,
$\|y_x\|_\sigma\le 1$. The following is an easy observation  from  the definition of the
injective norm.

\begin{prop}\label{prop:injective} 
Let  $y\in V^g$. Then $\{\frac12(\sigma\pm y_x)\}$ is an assemblage if and only if 
$\|y\|_{\epsilon,\sigma}\le 1$, where $\|\cdot\|_{\epsilon,\sigma}$ is the injective cross  norm
in the tensor product $\ell_\infty^g\otimes (V,\|\cdot\|_\sigma)$.

\end{prop}

\begin{proof} This follows from
\[
\|y\|_{\epsilon,\sigma}=\max_x \|y_x\|_\sigma.
\]

\end{proof}

%

We now  define a new norm in $\ell_\infty^g\otimes V$: For $y=(y_1,\dots,y_g)\in V^g\simeq
\mathbb R^g\otimes V$,  put 
\[
\|y\|_{steer, \sigma}:=\inf\{ \|\sum_j \phi_j\|_\sigma,\ y=\sum_j z_j\otimes \phi_j,\
		\|z_j\|_\infty= 1,\ \phi_j\in V^+\}.
\]

\begin{prop}\label{prop:crossnorm} $\|\cdot\|_{steer,\sigma}$ is a reasonable  cross norm on $\ell_\infty^g\otimes
(V,\|\cdot\|_\sigma)$. 

\end{prop}

\begin{proof} It is easily seen that $\|\cdot\|_{steer,\sigma}$ is a norm, so it suffices
 to show that $\|\cdot\|_{\epsilon,\sigma}\le
\|\cdot\|_{steer,\sigma}\le \|\cdot\|_{\pi,\sigma}$, where $\|\cdot\|_{\pi,\sigma}$ is the
projective norm. Assume that $y=\sum_j z_j\otimes \phi_j$ with $z_j\in \mathbb R^g$, $\|z_j\|_\infty=1$ and
$\phi_j \in V^+$. Then $y_x=\sum_j z_{j,x}\phi_j$ so that 
\[
-\|\sum_j \phi_j\|_\sigma\sigma\le -\sum_j \phi_j\le  -\sum_j|z_{j,x}|\phi_j\le y_x\le
\sum_j|z_{j,x}|\phi_j\le \sum_j\phi_j\le \|\sum_j \phi_j\|_\sigma\sigma
\] 
hence $\|y_x\|_\sigma\le \|y\|_{steer,\sigma}$ for all $x$, this implies the first
inequality. For the second inequality, let $y=\sum_j \tilde z_j\otimes \psi_j$ with
$\tilde z_j\in
\ell_\infty^g$ and $\psi_j\in V$. Put $z_j=\|\tilde z_j\|_\infty^{-1}\tilde z_j$ and let
$\psi^\pm_j=\frac12(\|\psi_j\|_\sigma\sigma\pm \psi_j)\in V^+$. Then
$\psi_j=\psi_j^+-\psi_j^-$ 
and we have
\[
y=\sum_j z_j \|\tilde z_j\|_\infty\psi_j^++ \sum_j(-z_j)\|\tilde z_j\|_\infty\psi_j^-
\]
and 
\[
\|y\|_{steer,\sigma}\le \|\sum_j \|\tilde z_j\|_\infty(\psi_j^++\psi_j^-)\|_\sigma\le
\sum_j\|\tilde z_j\|_\infty\|\psi_j\|_\sigma.
\]
This implies the second inequality.

\end{proof}

\begin{thm}\label{thm:lhs_norm} Let $y\in V^g$ be such that $\{\frac12(\sigma\pm y_x)\}$
is an assemblage. Then the  assemblage is classical  if and only if 
$\|y\|_{steer,\sigma}\le 1$.

\end{thm}

\begin{proof} By Theorem \ref{thm:lhs}, the assemblage is classical if and only if
$(\sigma,y)\in S_g\otimesmin K$. It is immediate from \eqref{eq:cs_min} that this is
equivalent to $\|y\|_{steer,\sigma}\le 1$.

\end{proof}

%
%
%

\subsection{Steering witnesses}\label{sec:witnesses} Let $\xi\in S_\kk\otimesmax K$ be any element.  
It is clear from the above results and the definition of the minimal and maximal tensor
products that the corresponding assemblage of shape $\kk$  is classical if and only if 
$\<w,\xi\>\ge 0$ for any $w\in A_{S_\kk}^+\otimesmax A^+$.
Therefore  any  element $w\in A_{S_\kk}^+\otimesmax A^+$ defines  a steering witness.  We say that
a steering witness is strict if there
is some  $\xi\in S_\kk\otimesmax K$
such that $\<w,\xi\><0$, which means that $w\notin A_{S_\kk}^+\otimesmin A^+$. Hence the set of all strict steering
witnesses is $A_{S_\kk}^+\otimesmax A^+\setminus A_{S_\kk}^+\otimesmin A^+$.
From now on we  restrict to dichotomic ensembles.

\begin{prop}\label{prop:witness}
An element $(w_0,w_1,\dots,w_g)\in A^{g+1}\simeq \mathbb R^{g+1}\otimes A$ is a steering
witness if and only if 
\begin{equation}\label{eq:steering_witness}
\sum_{x=1}^g \varepsilon_x w_x\le w_0,\qquad \forall \varepsilon \in \{\pm 1\}^g.
\end{equation}
In this case, we always have $w_0\in A^+$ and $w_0=0$ implies  $w_x=0$ for all $x$. 
A steering witness is strict if and only if there exist some $\sigma\in K$ and elements
$y_x\in V$, $\|y_x\|_\sigma\le 1$ such that
\[
\<w_0,\sigma\>+\sum_x \<w_x,y_x\><0.
\]

\end{prop}

\begin{proof} The first part follows from  Example \ref{ex:cs}  and the fact that the  extremal points $\epsilon \in
[-1,1]^g$ are  are precisely the elements of $\{\pm1\}^g$.  The last statement follows from the representation of
$S_g\otimesmax K$  in the previous section.

\end{proof}

Let now $\sigma\in \mathrm{int}(V^+)$ and let us restrict to assemblages in $\mathcal
A_{2,\sigma}^g$. Let $w=(w_0,w_1,\dots,w_g)$ be a steering witness. Since $\sigma$
is an interior point, we have by Proposition \ref{prop:witness} that $\<w_0,\sigma\>>0$
unless $w_0=0$ and in this case $w=0$. Therefore we may restrict to witnesses with
$\<w_0,\sigma\>=1$ and then  the value of the witnesses on elements in $\mathcal
A_{2,\sigma}^g$ is determined by the $g$-tuple $(w_1,\dots,w_g)\in A^g$.

\begin{prop}\label{prop:2sigma_witnesses} Let $w\in A^g$. The following are
equivalent:
\begin{enumerate}
\item[(i)] there is some $w_0\in A^+$, such that  $\<w_0,\sigma\>=1$ and
$(w_0,w)$ is a steering witness. 
\item[(ii)] $\sum_{x=1}^g |\<w_x,y_x\>|\le 1$ for all $y\in V^g$ such that
$\|y\|_{steer,\sigma}\le 1$.

\end{enumerate}

\end{prop}

\begin{proof} Assume (i) and let $\|y\|_{steer,\sigma}\le 1$. Note that  by
definition of the norm, we have for any $\epsilon \in \{\pm 1\}^g$,
\[
\|(\epsilon_1y_1,\dots,\epsilon_gy_g)\|_{steer,\sigma}=\|(y_1,\dots,y_g)\|_{steer,\sigma},
\]
 so that  $(\sigma, \epsilon_1y_1,\dots,\epsilon_gy_g)$ corresponds to a classical  assemblage.
 Therefore we
 have 
 \[
 0\le \<(w_0,w_1,\dots,w_g), (\sigma,
 \epsilon_1y_1,\dots,\epsilon_gy_g)\>=1+\sum_x \epsilon_x \<w_x,y_x\>,\qquad \forall
 \epsilon \in \{\pm 1\}^g,
 \]
 this proves (ii).

For the converse, choose any element $w_0'\in A$ such that $\<w_0',\sigma\>=1$. Let
$(\sigma,y)\in S_g\otimesmin K$, then $\|y\|_{steer,\sigma}\le 1$ and  by (ii)
\[
\<(w'_0,w_1,\dots,w_g),(\sigma,y_1,\dots,y_g)\>=1+\sum_x\<w_x,y_x\>\ge 1-\sum_x|\<w_x,y_x\>|\ge 0.
\]
It follows that $w'=(w_0',w_1,\dots,w_g)$ defines a positive functional on the subspace 
\[
\mathcal L=\{(y_0,y_1,\dots,y_g)\in V^{g+1},\ y_0\in \mathbb R\sigma\}
\]
 with the cone $\mathcal L\cap (V_g^+\otimesmin V^+)$. Since this subspace contains  the
 interior element $(\sigma,0) \in
\mathrm{int}( V_g^+\otimesmin V^+)$, $w'$ extends to an element  $w\in A_g^+\otimesmax A^+$.
It is easily checked that $w=(w_0,w_1,\dots,w_g)$ with  $w_0\in A^+$ and  
\[
\<w_0,\sigma\>=\<w, (\sigma,0)\>=\<w',(\sigma,0)\>=\<w_0',\sigma\>=1.
\]
This finishes the proof.

\end{proof}

We will denote the set of $w\in A^g$ satisfying the above conditions by
$\mathcal W_{2,\sigma}^g$, note that this is the unit ball in $\ell_1^g\otimes A$ with
respect to the cross norm dual to $\|\cdot\|_{steer,\sigma}$. By the above result, up to
multiplication by a positive constant,
$\mathcal W_{2,\sigma}^g$ is  the set of steering witnesses for assemblages in $\mathcal
A_{2,\sigma}^g$. We will further say that such a witness is strict if it has a negative
value on some assemblage in $\mathcal A_{2,\sigma}^g$. For the following characterization
of strict witnesses,
recall that  the base norm for the formally
introduced system $(A,A^+,\sigma)$ has the form
\[
\|h\|^\sigma= \sup_{\pm y\le \sigma} \<h,y\>.
\]

\begin{prop} The steering witness $w\in \mathcal W_{2,\sigma}$ is
strict if and only if
\[
\|w\|_{\pi,\sigma}=\sum_{x=1}^g\|w_x\|^\sigma>1,
\]
here $\|\cdot\|_{\pi,\sigma}$ is the projective cross  norm in the tensor product
$\ell_1^g\otimes (A,\|\cdot\|^\sigma)$.

\end{prop}

\begin{proof} It is easy to check directly that $\|w\|_{\pi,\sigma}=\sum_{x=1}^g\|w_x\|^\sigma$. By Proposition \ref{prop:witness}, 
 $w$ is a strict steering witness  if and only if there are some
 $y_x\in V$, $\|y_x\|_\sigma\le 1$ such that 
 \[
0> 1+\sum_x \<w_x,y_x\>= 1-\sum_x\<w_x,-y_x\>\ge 1-\sum_x\|w_x\|^\sigma.
\]

\end{proof}

\subsection{Steering degree}\label{sec:degree}
The steering degree of assemblages of shape $\kk$ can be quantified by the amount of noise that needs
to be mixed with the assemblage $\{\rho_{a|x}\}$ in order to
obtain a classical  assemblage, this is also called steering robustness
\cite{piani2015necessary}. The noise is represented by
assemblages of the same shape as $\{\rho_{a|x}\}$, see \cite{cavalcanti2016quantum,
cavalcanti2016quantitative} for some variants. Note that the convex
structure on the set of assemblages is inherited from $S_\kk\otimesmax K$. 
Here we will use a
single trivial assemblage, of the form
 $\{\omega_{a|x}\sigma\}$ with  $\omega_{a|x}=|\Omega_x|^{-1}$ for all $a\in \Omega_x$. The
 steering degree of $\{\rho_{a|x}\}$ is defined as
 \[
s(\{\rho_{a|x}\})=\sup\{s\in [0,1], \ \{s\rho_{a|x}+(1-s)|\Omega_x|^{-1}\sigma\} \text{ is
classical}\}.
 \]

We also define $s_{\kk,\sigma}(K)$ as the infimum of the steering degrees of all
assemblages with  shape $\kk$ and barycenter $\sigma$. 
 Since we assume $K$ fixed, we will skip it from the notation.
Restricting to dichotomic assemblages, we now show that the steering degree can be
expressed using the norm 
$\|\cdot\|_{steer,\sigma}$. In this case we denote $s_{\kk,\sigma}\equiv s_{g,\sigma}$.

\begin{thm}\label{thm:sdegree} Let $\{\rho_{\pm|x}\}\in \mathcal A_{2,\sigma}^g$ and let
$y_x=\rho_{+|x}-\rho_{-|x}$, $x=1,\dots,g$. Then
\begin{equation}\label{eq:sdegree}
s(\{\rho_{\pm|x}\})=\|(y_1,\dots,y_g)\|_{steer,\sigma}^{-1}.
\end{equation}
For the overall steering degree, we have
\begin{equation}\label{eq:sdegree_all}
s_{g,\sigma}=\sup_{y\in \mathbb
R^g\otimes V}\frac{\|y\|_{\epsilon,\sigma}}{\|y\|_{steer,\sigma}}\ge \sup_{y\in \mathbb
R^g\otimes V}\frac{\|y\|_{\epsilon,\sigma}}{\|y\|_{\pi,\sigma}}.
\end{equation}
Dually, in terms of the steering witnesses, we obtain
\begin{equation}\label{eq:sdegree_wit}
s_{g,\sigma}=\sup\{s\in [0,1],\ \sum_{x=1}^g s\|w_x\|^{\sigma}\le 1,\forall 
(w_1,\dots, w_g)\in \mathcal W_{2,\sigma}^g\}.
\end{equation}

\end{thm}

\begin{proof}
Let $s\in [0,1]$, then the tensor element corresponding to the mixed assemblage
$\{s\rho_{\pm|x}+(1-s)\frac12 \sigma\}$ is $(\sigma,sy_1,\dots,sy_g)$. By Theorem
\ref{thm:lhs_norm} we have
\[
s(\{\rho_{\pm|x}\})=\sup\{s\in [0,1],\ s\|(y_1,\dots,y_g)\|_{steer,\sigma}\le 1\}= \|(y_1,\dots,y_g)\|_{steer,\sigma}^{-1}.
\]
The equality \eqref{eq:sdegree_all} now follows from Proposition \ref{prop:injective} and
Proposition \ref{prop:crossnorm},
\eqref{eq:sdegree_wit} follows from Proposition \ref{prop:2sigma_witnesses} and the
definition of the norm $\|\cdot\|^\sigma$.

\end{proof}

In Section \ref{sec:co_sdegree} below we  will  characterize the universal  steering
degree for dichotomic assemblages with barycenter $\sigma$: 
\begin{align}
s_\sigma&:=\max\{s\in [0,1],\ (\sigma,sy)\in S_g\otimesmin K,\ \forall (\sigma,y)\in
S_g\otimesmax K, \ \forall g\} \\ \label{eq:ssigma_witness}
        &=\max\{s\in [0,1],\ s\sum_i \|w_i\|^\sigma\le 1,\ \forall w\in
	\mathcal W_{2,\sigma}^g,\forall g\}\\
&=\inf_{g\in \mathbb N} s_{g,\sigma}.
\end{align}

\subsection{Relation to compatibility of measurements}

Let $\{f_x\}_{x=1}^g$ be a collection of measurements with outcomes in $\Omega_x$, with effects
$f_{a|x}$. We say that the collection is compatible if all $f_x$ are marginals of a joint
measurement $h$ with outcomes in $Y=\Pi_x \Omega_x$:
\[
f_{a|x}=\sum_{(a_1,\dots,a_g)\in Y, a_x=a} h_{a_1,\dots,a_g},\qquad \forall a, x.
\]
The aim of the present section is to remark that such collections of measurements and
existence of a joint measurement for them is mathematically   equivalent to assemblages
and existence of LHS models. This simple
observation shows a link to the previously obtained results of a relation of 
incompatibility and related notions of witnesses and degree to minimal/maximal tensor products of cones and tensor norms, obtained
in \cite{bluhm2020incompatibility}.

So let $\sigma$ be an arbitrary interior point in the state space $K$. As before, we may
formally consider the system
$(A,A^+,\sigma)$, with the state space $K^\sigma\subseteq A^+$. Note that we always have $\ou\in K^\sigma$. 
The following  is rather straightforward.
\begin{lemma}\label{lemma:meas_assemb}
 Sets of measurements on  $(V,V^+,\ou)$ correspond precisely to 
assemblages for $(A,A^+,\sigma)$ with barycenter $\ou$.  Moreover, the measurements are compatible if and only
if the corresponding assemblage is classical.
\end{lemma}

Note that the norm $\|\cdot\|_\ou$ on $A$ is the usual order unit norm and the dual norm
$\|\cdot\|^\ou$ in $V$  is the
base norm $\|\cdot\|_V$.  The results of the previous sections  correspond to the results
in  \cite{bluhm2020incompatibility} 
for compatibility of  dichotomic measurements, in particular the norm $\|\cdot\|_{steer,\ou}$ on
$\ell_\infty^g\otimes A$ becomes precisely the compatibility norm $\|\cdot\|_{c}$ as in
\cite{bluhm2020incompatibility}. 

The above relations have some immediate consequences, obtained from the results in
\cite{bluhm2020incompatibility} and duality relations:
\begin{enumerate}
\item The steering degree for any
$(V,V^+,\ou)$ with $\dim(V)=d$ is lower bounded by 
\[
s_{\sigma,g}\ge 1/\min\{g,d\}.
\]
\item\label{itm:QMST_inc} In the quantum case, the
assemblages are directly related to collections of measurements by the map $\rho\mapsto
\sigma^{-1/2}\rho\sigma^{-1/2}$, mapping classical assemblages onto compatible
measurements and relating the steering degrees to the compatibility degrees,
\cite{uola2015one}. Note that if
the barycenter $\sigma$ is not a faithful state then we obtain measurements for a quantum
system of lower dimension.
\item  In the centrally symmetric case,  we have for $\sigma=(1,0)$  the tight lower
bound 
\[
s_{(1,0),g}\ge 1/\min\{g,d-1\},
\]
attained for the state spaces isomorphic to
cross-polytopes (unit balls of the $\ell_1$-norm). In general, for a norm $\|\cdot\|$
in $\mathbb R^n$ and its unit ball $B$, the steering degrees $s_{\kk,(1,0)}(B)$ are the
same as the compatibility degrees for the dual unit ball $B^*$.

\end{enumerate}

\section{Steering and Choquet order}\label{sec:choquet}

In this section, we show how  the properties of probability measures on compact convex
sets can be used to characterize steering in GPTs. Although for simplicity the dimension
of the systems is assumed to be finite, we remark that most of the results hold as stated here
for metrizable convex compacts and with slight technical modifications also for arbitrary
compact convex subsets of a locally convex space.

Let $(V,V^+,\ou)$ be a system  with state space $K$.
Let $C(K)$ denote the Banach space of continuous functions $f:K\to \mathbb R$
with maximum norm. The dual space  $(A,A^+)$ to $(V,V^+)$ can be identified with the subspace $A(K)\subseteq C(K)$ of  affine functions over $K$, with the cone of
positive functions $A^+=A(K)^+$, the order unit $\ou$ is the constant unit functional over $K$. The order unit norm in $A$
coincides with the maximum norm on $A(K)$. Let us denote by $P(K)\subseteq C(K)$ the cone of convex functions in $C(K)$.

 The dual space $C^*(K)$ is the space of signed Radon measures over $K$. Let us denote by  $\mathcal P(K)$
the set of Radon probability measures over $K$. Then $\mathcal P(K)$ is compact in the weak*-topology inherited from the
Banach space  duality with  $C(K)$, in fact,  $\mathcal P(K)$ is a Choquet simplex. 
For any $\sigma\in K$, the probability measure
concentrated in $\sigma$ belongs to $\mathcal P(K)$ and is denoted by $\delta_\sigma$. A
measure of the form $\sum_{i=1}^n c_i \delta_{\rho_i}$, $\rho_i\in K$, $c_i\in
\mathbb R$ is called simple. For $\sigma\in K$, we denote by $\mathcal
P_\sigma(K)$ the subset of probability measures $\mu\in \mathcal P(K)$ with barycenter
$\bar{\mu}=\int_K \rho d\mu(\rho)=\sigma$.

Recall that  the Choquet order on the set of Radon  measures over $K$ is defined as the dual of the ordering in $C(K)$ obtained
from the cone $P(K)$: we have $\nu\prec\mu$ if 
\[
\int fd\nu\le \int fd\mu,\qquad \forall f\in P(K).
\] 
 If $\nu \prec \mu$, then 
 $\bar\mu=\bar \nu$. A positive measure is maximal with respect to this ordering if and
 only if it 
is a boundary measure, that is,  concentrated on the extreme boundary
$\partial_e K$. We will denote the set of all boundary measures in $\mathcal P_\sigma(K)$ by
$\mathcal P_\sigma^b(K)$. Further, any positive measure is upper bounded by a positive boundary measure, in particular, any
element $\sigma\in
K$ is the barycenter of some measure $\mu \in \mathcal P^b_\sigma(K)$. For details
see e.g.
\cite{alfsen1971compact, phelps2001lectures}.

\subsection{Boundary measures and LHS models} Let $\{\rho_{a|x}\}$ be an assemblage with
barycenter $\sigma$. 
Observe that a LHS model of the form \eqref{eq:unsteerable} for   $\{\rho_{a|x}\}$  can be expressed as
\[
\rho_{a|x}= \int_K q(a|x, \rho)\rho d\mu(\rho),
\]
where $\mu=\sum_{\lambda\in \Lambda} q(\lambda)\delta_{\rho_\lambda}\in \mathcal
P_\sigma(K)$.

We will now identify $\{\rho_{a|x}\}$ with a set of simple probability measures in  $P_\sigma(K)$.
Put $\lambda_{a|x}:=\<\ou,\rho_{a|x}\>$ and $\sigma_{a|x}:=\lambda^{-1}_{a|x}\rho_{a|x}\in K$ 
(if $\lambda_{a|x}=0$ we may pick any state
$\sigma_{a|x}\in K$). For any $x$, put
\[
\mu_x:= \sum_{a\in \Omega_x} \lambda_{a|x}\delta_{\sigma_{a|x}}\in \mathcal P(K),\quad
\bar{\mu}_x=\sum_a \lambda_{a|x}\sigma_{a|x}=\sigma,
\]
so that we can represent the assemblage as a (finite) set
$\{\mu_x\}_{x=1}^g\subseteq \mathcal P_\sigma(K)$.  Conversely, for any such subset,
 $\{\lambda_{a|x}\sigma_{a|x}\}$ is an assemblage with barycenter
$\sigma$. Note that this representation, in contrast with the identification with
$S\otimesmax K$ in the previous section, ignores any permutation of the measurements or
relabelling of the outcomes, but in the context of steering these are irrelevant.
 Note also that the
convex structure induced from $\mathcal P_\sigma(K)$ through this representation is
different from the one used in the previous paragraphs. We
next show that existence of a LHS model can be expressed in terms of the Choquet order in
$\mathcal P_\sigma(K)$.

\begin{prop} \label{prop:simple_co} Let $\nu\in \mathcal P_\sigma(K)$ be a simple measure,
$\nu=\sum_a\lambda_a \delta_{\sigma_a}$. Then $\nu\prec \mu$ for some $\mu\in \mathcal
P_\sigma(K)$ if and only if there are measurable functions $q(a|\cdot): K\to [0,1]$ such
that $\sum_a q(a|\rho)=1$ and 
\[
\lambda_a \sigma_a= \int_K \rho q(a|\rho)d\mu(\rho).
\]

\end{prop}

\begin{proof} Assume that $\nu\prec \mu$. Since $\nu$ is simple, by 
\cite[Cor. I.3.4]{alfsen1971compact} this means that there is a convex decomposition 
 $\mu=\sum_a \lambda_{a} \mu_{a}$ such that $\mu_{a}\in \mathcal P_{\sigma_{a}}(K)$. 
 Let $f(a|\cdot)=\frac {d\mu_a}{d\mu}$   and put $q(a|\cdot )=\lambda_af(a|\cdot)$, then
 we may assume that $q(a|\rho)\in [0,1]$ and $\sum_a q(a|\rho)=1$ by suitably replacing
 the values of the functions $f(a|\cdot)$ on a subset $K_0\subseteq K$ with $\mu(K_0)=0$.
 Then
 \[
\lambda_a\sigma_a=\lambda_a\bar\mu_a=\lambda_a\int_K \rho d\mu_a(\rho)=\int_K\rho
q(a|\rho)d\mu(\rho).
 \]

Assume the converse, then $\lambda_a=\int_K q(a|\rho)d\mu$ so that
$\lambda^{-1}_aq(a|\cdot)d\mu$ defines a probability measure $\mu_a$ such that
$\sigma_a=\int \rho d\mu_a(\rho)$. For any $f\in P(K)$ we have
\[
\int f d\nu=\sum_a \lambda_a f(\sigma_a)\le \sum_a\lambda_a \int_K f(\rho)d\mu_a=\sum_a
\int_K f(\rho)q(a|\rho)d\mu=\int_K fd\mu.
\]

\end{proof}

We now extend the definition of an assemblage as an arbitrary set $\{\mu_x\}_{x\in X}$ of simple measures with a common barycenter, so that we no longer assume  the
parameter set $X$ to be finite. We say that an assemblage
$\{\sum_a\lambda_{a|x}\delta_{\sigma_{a|x}}\}_{x\in X}\subseteq \mathcal P_\sigma(K)$ admits
a LHS model (or is classical) if there is some $\mu\in \mathcal P_\sigma(K)$ and measurable functions
$q(a|x,\cdot): K\to [0,1]$, $\sum_a q(a|x,\cdot)=1$ for all $x\in X$, such that 
\begin{equation}\label{eq:LHS_inf}
\lambda_{a|x}\sigma_{a|x}=\int_K \rho q(a|x,\rho)d\mu(\rho),\qquad a\in \Omega_x,\ x\in X.
\end{equation}

The following theorem  collects some observations for this definition of a local hidden
state model and its relation to Choquet order. Note that the statement (iii)
shows that for finite assemblages this definition of a LHS model coincides with the
previous one from Section \ref{sec:lhs}.

\begin{thm}\label{thm:lhs_co}  Let $\{\mu_x\}_{x\in X}\subseteq \mathcal
P_\sigma(K)$ be an assemblage. Then
\begin{enumerate}
\item[(i)] The assemblage is classical if and only if all the measures $\mu_x$
have a common upper bound in Choquet order. In this case, any measure $\mu$ such that $\{\mu_x\}\prec
\mu$ defines some LHS model.
\item[(ii)] If $\{\mu_x\}\prec \mu$, then we may always assume that $\mu\in \mathcal
P_\sigma^b(K)$.
\item[(iii)] If $X$ is a finite set and $\{\mu_x\}\prec \mu$, then we may assume
that $\mu$ is simple.

\item[(iv)] The assemblage is classical  if and only if any finite sub-assemblage
$\{\mu_x\}_{x\in F}$, $F\subseteq X$, $|F|<\infty$ is
classical.
\item[(v)] Assume that  $\{\mu_x\}\prec \mu$. If the assemblage is invariant under an affine bijection $T:K\to K$
(that is, for all $x\in X$, $\mu_x^T=\mu_{x'}$ for some $x'\in X$),
then we may assume that $\mu$ is invariant under $T$. 

\end{enumerate}

\end{thm}

\begin{proof}
The statement (i) follows immediately from Proposition \ref{prop:simple_co}, (ii) follows
from the fact that any measure is upper bounded (in the Choquet order) by a boundary
measure. For (iii), let $\{\mu_x=\sum_x \lambda_{a|x}\delta_{\sigma_{a|x}}\}_{x=1}^g$ and 
assume a LHS model \eqref{eq:LHS_inf} for the assemblage with some measure $\mu_0$. As in
the proof of Proposition \ref{prop:simple_co}, let
$\mu_{a|x}=\lambda_{a|x}^{-1}q(a|x,\cdot)d\mu_0\in \mathcal P_{\sigma_{a|x}}(K)$.  Then for
all $x=1,\dots,g$ we have a convex decomposition 
$\sum_a \lambda_{a|x}\mu_{a|x}=\mu_0$. Since $\mathcal P(K)$ is a Choquet simplex, all the decompositions have a common refinement: there 
 are probability measures $\mu_\omega$ indexed by $\omega\in \Omega=\Omega_1\times\dots\times
 \Omega_g$ and some $q\in \mathcal P(\Omega)$ such that $\mu_0=\sum_\omega
 q(\omega)\mu_\omega$ and 
\begin{align*}
\lambda_{a|x}\mu_{a|x}&=\sum_{\omega, \omega_x=a}
q(\omega)\mu_{\omega} = \sum_\omega d(a|x,\omega)q(\omega)\mu_\omega,
\end{align*}
where $d(a|x,\omega)=1$ if $\omega_x=a$ and is 0 otherwise.
Put $\rho_\omega:=\bar{\mu}_\omega$, then we obtain 
\[
\rho_{a|x}=\lambda_{a|x}\sigma_{a|x}=\lambda_{a|x}\bar\mu_{a|x}=\sum_\omega
d(a|x,\omega)q(\omega)\rho_\omega,
\]
which is a LHS model with a simple measure $\mu:=\sum_\omega
q(\omega)\delta_{\rho_\omega}$.
 
To prove (iv), assume that any finite sub-assemblage $\{\mu_x\}_{x\in F}$ is
classical. By (i), this is equivalent to the fact that for any finite $F\subseteq X$, the subset 
\[
M_{ F}:=\{\mu\in \mathcal P_\sigma(K),\ \mu_x\prec \mu,\
\forall x\in  F\}
\]
is nonempty and it is easily seen from the definition of Choquet order that $M_{
F}$ is also closed, in the topology of $\mathcal P(K)$. Moreover, since for any finite
collection $F_i\subseteq \mathsf A$, $i=1,\dots,n$, we have $\bigcap_i M_{F_i}=M_{\cup_i F_i}$, we see that 
\[
\{M_{F},\ F\subseteq X\mbox{ is finite}\}
\]
is a collection of closed subsets in $\mathcal P(K)$ with the finite intersection
property. The statement now follows by compactness of $\mathcal P(K)$.

To prove (v), let $\{\mu_x\}\prec \mu$ and let $T:K\to K$ be an  affine bijection preserving $\{\mu_x\}_{x\in  X}$,
then for any $f\in P(K)$ and $x\in  X$,
\[
\int_K fd\mu_x=\int_K f(T^{-1}(\rho))d\mu_x^{T}(\rho)\le \int_K f\circ
T^{-1}d\mu=\int_K fd\mu^T
\]
so that $\mu_x\prec \mu^T$.  The set of all affine bijections
$K\to K$ form a compact group (in the
topology of pointwise convergence) of which the elements preserving the assemblage form a
compact subgroup $G$. It is easily checked that the map $G\to \mathcal P(K)$ given by 
$S\mapsto \mu^S$ is continuous. Let  $m$ be the Haar measure for $G$  and let
$\mu_{m}=\int_G \mu^Sdm(S)$, then $\mu_m$ is invariant under $T$ and we have for any $f\in
P(K)$ and $x\in  X$:
\[
\int_K fd\mu_m=\int_G \int_K f(\rho)d\mu^S(\rho)dm(S)\ge \int fd\mu_x.
\]
\end{proof}

We now give some further characterization  of the Choquet order in the case of simple
measures.

\begin{prop}\label{prop:simple_co_dual} Let $\mu,\nu\in \mathcal P(K)$ and assume that
$\nu=\sum_{a=1}^k\lambda_a\delta_{\sigma_a}$ is simple. Then $\nu\prec \mu$
if and only if for all $g_1,\dots,g_k\in A$ we have
\[
\sum_a \lambda_a \<g_a,\sigma_a\>\le \int_K (g_1\vee\dots\vee g_k)(\rho)d\mu.
\]

\end{prop}

\begin{proof} Assume the inequality holds for all $g_1,\dots,g_k\in A$. Let $f\in P(K)$,
then there are affine functions $g_a\in A(K)=A$, $a=1,\dots,k$,  such that $g_a\le f$ and
$f(\sigma_a)=g_a(\sigma_a)$. Then we have  $\bigvee_{a'}g_{a'}\le f$
and therefore
\begin{align*}
\int fd\nu&=\sum_a \lambda_a f(\sigma_a)=\sum_a\lambda_a \<g_a,\sigma_a\>\le 
 \int_K
(\bigvee_{a'} g_{a'})(\rho)d\mu\le \int_K fd\mu.
\end{align*}
For the converse, assume that $\nu\prec \mu$, then
\[
\sum_a \lambda_a\<g_a,\sigma_a\>\le \sum_a \lambda_a \<\bigvee_{a'}g_{a'},\sigma_a\>=
\int (\bigvee_{a'}g_{a'})d\nu\le \int (\bigvee_{a'}g_{a'})d\mu,
\]
the last inequality follows from the fact that the maximum of affine functions is
convex.

\end{proof}

\subsection{Dichotomic assemblages and steering degree}\label{sec:co_sdegree}

We now restrict our attention to dichotomic assemblages
$\{\mu_{x}\}_{x\in X}\subseteq \mathcal P_{2,\sigma}(K)$, where $\mathcal
P_{2,\sigma}(K)$ denotes the subset of measures in $\mathcal P_\sigma(K)$ supported in two points. 
In this case we obtain a simpler characterization of the Choquet order.


\begin{lemma}\label{lemma:dichotomic_co} For $\nu\in \mathcal P_{2,\sigma}(K)$ and $\mu\in \mathcal P_{\sigma}(K)$, we have
$\nu\prec \mu$ if and only
if 
\[
\int_K |\<h,\rho\>|d\nu(\rho)\le \int_K |\<h,\rho\>|d\mu(\rho),\qquad \forall h\in A.
\]
\end{lemma}

\begin{proof} Let $g_\pm\in A$, $\rho\in K$. Note that
\[
(g_+\vee g_-)(\rho)=\max\{\<g_+,\rho\>,\<g_-,\rho\>\}=\frac12 (|\<g_+-g_-,\rho\>|+
\<g_++g_-,\rho\>).
\]
If the inequality in the lemma  is satisfied, then we have
\begin{align*}
\lambda_+\<g_+,\sigma_+\>+\lambda_-\<g_-,\sigma_-\>&\le \int_K (g_+\vee g_-)(\rho)d\nu= \frac12\biggl(\int_K|\<g_+-g_-,\rho\>|d\nu(\rho)+
\<g_++g_-,\sigma\>\biggr) \\
&\le \frac12\biggl(\int_K|\<g_+-g_-,\rho\>|d\mu(\rho)+
\<g_++g_-,\sigma\>\biggr)
=\int_K (g_+\vee g_-)(\rho)d\mu
\end{align*}
By Proposition \ref{prop:simple_co_dual}, this implies  $\nu\prec \mu$. The converse holds
since $\rho\mapsto |\<h,\rho\>|$ is convex for any $h\in A$.

\end{proof}

\begin{lemma}\label{lemma:co_norm} For any $h\in A$ and $\mu\in \mathcal P_\sigma(K)$, we have
\[
\int_K |\<h,\rho\>|d\mu\le \|h\|^{\sigma}=\max_{\nu\in \mathcal P_{2,\sigma}(K)}\int_K |\<h,\rho\>|d\nu.
\]

\end{lemma}

\begin{proof} Let $K_\pm=\{\rho\in K, \pm\<h,\rho\>\ge 0\}$, then 
\[
\int_K|\<h,\rho\>|d\mu(\rho)=\int_{K_+} \<h,\rho\>d\mu-\int_{K_-}\<h,\rho\>d\mu=
\<h,\mu_+-\mu_-\>
\]
where $\mu_\pm=\int_{K_\pm}\rho d\mu(\rho)$. Since $\mu_++\mu_-=\sigma$, we have
$\pm(\mu_+-\mu_-)\le \sigma$, so that $\<h,\mu_+-\mu_-\>\le \|h\|^\sigma$.
Since $\|\cdot\|^{\sigma}$ is a base norm with respect to the order unit $\sigma$, there are some $y_\pm \in V^+$ such
that $y_++y_-=\sigma$ 
and 
\[
\|h\|^{\sigma}=\<h,y_+-y_-\>=|\<h,y_+\>|+ |\<h,y_-\>|=\int_K |\<h,\rho\>|d \nu(\rho),
\] 
where $ \nu\in \mathcal P_{2,\sigma}(K)$.

\end{proof}

We also have  an alternative characterization of the witness set $\mathcal
W_{2,\sigma}^g$.

\begin{lemma}\label{lemma:witnesses} Let $w\in A^g$. Then $w\in
\mathcal W_{2,\sigma}^g$ if and
only if  for all $\mu\in \mathcal
P_\sigma(K)$,
\[
\sum_x\int_K  |\<w_x,\rho\>|d\mu(\rho)\le 1.
\]
\end{lemma}

\begin{proof} Assume that $(w_1,\dots, w_g)\in \mathcal W_{2,\sigma}^g$, then by
Proposition \ref{prop:witness} there is some $w_0\in A^+$, $\<w_0,\sigma\>=1$ such that 
$\sum_x\epsilon_xw_x \le w_0$ for all $\epsilon \in \{\pm1\}^g$. For any $\rho\in K$, there
is some $\epsilon \in \{\pm1\}^g$ such that
\[
\sum_x|\<w_x,\rho\>|=\sum_x\epsilon_x \<w_x,\rho\>\le \<w_0,\rho\>.
\]
For any $\mu\in \mathcal P_\sigma(K)$ we obtain
\[
\sum_x \int_K|\<w_x,\rho\>|d\mu\le \int_K\<w_0,\rho\>d\mu=\<w_0,\sigma\>=1.
\]

For the converse, let
$\{\mu_x=\lambda\delta_{\sigma_{+|x}}+(1-\lambda)\delta_{\sigma_{-|x}}\}_{x=1}^g$ be a
classical
dichotomic assemblage. Note that the corresponding element in
$S_g\otimesmin K$ has the form $(\sigma,y)$ with
$y_x=\lambda\sigma_{+|x}-(1-\lambda)\sigma_{-|x}$. Let $\mu\in \mathcal P_\sigma(K)$ be
such that $\mu_x\prec\mu$ for $x=1,\dots,g$, then by the
definition of Choquet order 
\[
\sum_x |\<w_x,y_x\>|\le \sum_x
\lambda|\<w_x,\sigma_{+|x}\>|+(1-\lambda)|\<w_x,\sigma_{-|x}\>|=
\sum_x \int_K |\<w_x,\rho\>|d\mu_x\le \sum_x \int_K |\<w_x,\rho\>|d\mu\le 1.
\]
The assertion now follows from Prop. \ref{prop:2sigma_witnesses}.

\end{proof}

We now obtain an expression for the  universal  steering degree $s_\sigma$.

\begin{thm}\label{thm:sdegree} For  $\mu\in \mathcal P_\sigma(K)$, let 
\[
c_\mu:=\inf_{h\in A,\|h\|^{\sigma}=1} \int_K |\<h,\rho\>|d\mu(\rho).
\]
Then $c_\mu\le s_\sigma$. There exists a boundary measure $\mu\in \mathcal
P_\sigma(K)$ such that $c_\mu=s_\sigma$, invariant under any affine bijection $K\to K$
that preserves $\sigma$.

\end{thm}

\begin{proof}
Let $\mu\in \mathcal P_\sigma(K)$. Note that $c_\mu$ is the largest $c\in [0,1]$ such that 
 $c \|h\|^{\sigma}\le \int_K |\<h,\rho\>|d\mu(\rho)$, for all $h\in A$. 
Let $(w_1,\dots,w_g)\in
\mathcal W_{2,\sigma}^g$, then by Lemma \ref{lemma:witnesses} we have
\[
c_\mu\sum_x\|w_x\|^\sigma\le \sum_x\int_K|\<w_x,\rho\>|d\mu(\rho)\le 1,
\]
so that $c_\mu\le s_\sigma$ by \eqref{eq:ssigma_witness}. 

We now prove existence of the measure such that equality is attained. Let $s=s_\sigma$ and
let $\nu=\lambda\delta_{\sigma_+}+(1-\lambda)\delta_{\sigma_-}\in \mathcal P_{2,\sigma}(K)$. 
Put $\rho_{+|\nu}=\lambda\sigma_+$, $\rho_{-|\nu}=(1-\lambda)\sigma_-$ and let 
$\mu_{s,\nu}\in \mathcal P_{2,\sigma}(K)$ be the measure corresponding to
$\{s\rho_{\pm|\nu}+(1-s)\frac12\sigma\}$. By definition of $s_\sigma$ and  Theorem
\ref{thm:lhs_co} (iv),  we see that $\{\mu_{s,\nu}\}_{\nu\in \mathcal P_{2,\sigma}(K)}$ is a
classical dichotomic assemblage, moreover, it is clearly invariant 
under the group $G_\sigma$ of  affine bijections that preserve $\sigma$. Using again Theorem \ref{thm:lhs_co},
we see that there is some  measure $\mu\in \mathcal P^b_\sigma(K)$, invariant under
$G_\sigma$ and such that
\[
\mu_{s,\nu}\prec \mu,\qquad \forall \nu\in \mathcal P_{2,\sigma}(K).
\]
Let $h\in A$ and assume that  $h\notin \pm A^+$. Then there are some $\rho_{\pm}\in V^+$,
$\rho_++\rho_-=\sigma$ such that $\|h\|^\sigma= |\<h,\rho_+\>|+|\<h,\rho_-\>|=\int
|\<h,\rho\>|d\nu$ and
$\pm\<h,\rho_\pm\>\ge 0$, here  $\nu\in \mathcal P_{2,\sigma}(K)$ is  the corresponding
measure. We then have
\[
\int_K |\<h,\rho\>|d\mu\ge \int_K |\<h,\rho\>|d\mu_{s,\nu}=\<h,s(\rho_+-\rho_-)\>=s\|h\|^\sigma
\]
If  $h\in \pm A^+$, then 
\[
\int_K |\<h,\rho\>|d\mu=|\<h,\sigma\>|=\|h\|^\sigma\ge s\|h\|^\sigma.
\]
It follows that $c_\mu\ge s$. By the first part of the proof, we now have
$c_\mu=s_\sigma$.

\end{proof}

\begin{ex}\label{ex:quantum_s}
For a quantum system $(M_d^{sa},M_n^+,I)$,  $s_\sigma$ is the same for all faithful states
$\sigma$. This follows by using the map $\sigma^{-1/2}\cdot \sigma^{-1/2}$ in point
\eqref{itm:QMST_inc} on p. \pageref{itm:QMST_inc}
to map assemblages to measurements  and
the relation to compatibility degree, see also \cite{bluhm2021maximal}. So we may choose $\sigma = n^{-1} I$ and then
$\|h\|^\sigma=n^{-1}\|h\|_\Tr$, where $\|h\|_\Tr$ is the trace norm. The universal
dichotomic steering degree (and
compatibility degree) is then given by
\[
s_\sigma=\inf_{\|h\|_\Tr=n} \int_{\mathcal P_n} |\<\psi,h\psi\>|d\mu(\psi),
\] 
where $\mu$ is the unique unitarily invariant measure on the set $\mathcal P_n$ of pure
states. This corresponds to the results of \cite{bluhm2021maximal}, where the infimum was
also evaluated.  

\end{ex}

\begin{ex}\label{ex:cs_s} Let $(V,V^+,\ou)$ be the centrally symmetric system given by the norm $\|\cdot\|$ in
$\mathbb R^n$ and let $\sigma$ be the central element $\sigma=(1,0)$. Note that in this
case the dichotomic assemblages are the same as sets of dichotomic measurements for 
the centrally symmetric system given by the norm $\|\cdot\|^*$, so by the results of
\cite{bluhm2020incompatibility} we have that $s_{(0,1)}=\pi_1^{-1}$ where $\pi_1$ is the 1-summing
constant for the norm $\|\cdot\|^*$. 

We now check this in our setting.  The base $K^\sigma$  of $A^+$ is the dual state space isomorphic to the unit ball of the
dual norm
\[
K^{(1,0)}=K^*=\{(1,\psi),\ \|\psi\|^*\le 1\}
\]
and the base norm is then $\|(t,\varphi)\|^{(1,0)}=\max\{|t|,\|\varphi\|^*\}$.
The boundary measures $\nu\in \mathcal P^b_{(1,0)}(K)$ which are invariant under affine bijections
preserving $(1,0)$  correspond to regular Borel probability measures on the set $C=\overline{\partial_eB}$
that are invariant under isometries of $\|\cdot\|$. Let $(t,\varphi)\in A$, $\|(t,\varphi)\|^{(1,0)}=1$, then
for any such $\nu$,
\begin{align*}
1&\ge \int_K |\<(t,\varphi),(1,x)\>|d\nu=\int_B|t+\<\varphi,x\>|d\nu=\frac12 \int_B
|t+\<\varphi,x\>|+ |t-\<\varphi,x\>|d\nu\\
&=\int_B \max\{|t|,|\<\varphi,x\>|\}d\nu\ge \int_B
|\<\varphi,x\>|d\nu=\int_K|\<(0,\varphi),(1,x)\>|d\nu
\end{align*}
(the second equality holds since $\nu$ is invariant under the map $x\mapsto -x$). If
$\|\varphi\|^*\le |t|$,
then we have $|\<\varphi,x\>|\le |t|$ for all $x\in B$ and the integral is equal to $|t|=1$. If
$\|\varphi\|^*\ge |t|$, then $\|\varphi\|^*=1$.
It follows that  infimum in the definition of $c_\nu$ is attained at an element with $t=0$ and we have
\[
c_\nu =\inf_{\|\varphi\|^*=1} \int_{C} |\<\varphi,x\>|d\mu\le \pi^{-1}_1.
\]
Here the  inequality
follows from \cite[Thm. 1]{gordon1969onpabsolutely}, moreover, equality is attained for
some invariant probability  measure $\nu_0$ on $C$.  It follows from our results that
there is some invariant boundary measure $\mu$ such that 
\[
\pi^{-1}_1=c_{\nu_0}\le s_{(1,0)}=c_\mu\le \pi^{-1}_1,
\]
so that, indeed,  $\pi^{-1}_1=s_{(1,0)}$.

\end{ex}

\section{Unsteerable states in GPTs}

Let $(V_A,V_A^+,\ou_A)$ and $(V_B,V_B^+,\ou_B)$ be two system and let $\sigma_{AB}\in
K_A\otimesmax K_B$. For any measurement $f$ on the system $V_A$, let $\nu_f\in \mathcal
P_{\sigma_B}(K_B)$ denote the
 simple measure given by the conditional states $\{(f_{a}\otimes id_B)(\sigma_{AB})\}$.
By Theorem \ref{thm:lhs_co}, $\sigma_{AB}$ is $(A\to B)$ unsteerable (by a
specified type of measurements) if and only if the assemblage of conditional states
$\{\nu_f\}$, parametrized by the set of all  measurements
$f$  (of the specified type) is classical: there is some measure $\mu \in \mathcal
P^b_{\sigma_B}(K_B)$ such that
for all $f$, 
\[
\{\nu_f\}\prec \mu.
\]
Assume that  $U_A:K_A\to K_A$ and $U_B:K_B\to K_B$ are  affine bijections such that
$U_A\otimes U_B$ preserves $\sigma_{AB}$. 
If also $U_A$ preserves the specified family of measurements, then since
\[
U_B((f_{a}\otimes id_B)(\sigma_{AB}))=(f_{a}\circ U_A^{-1}\circ U_A\otimes U_B)(\sigma_{AB})=
(f_{a}\circ U_A^{-1}\otimes id_B)(\sigma_{AB}),
\]
the assemblage is invariant under $U_B$, so we may assume that $\mu$ is an invariant
boundary probability measure. 

Note that the state $\sigma_{AB}$ determines an affine map  $S_{A\to B}: K_A^{\sigma_A}\to
K_B$ by
\[
\<h_A\otimes h_B,\sigma_{AB}\>=\<h_B,S_{A\to B}(h_B)\>,\qquad h_A\in K_A^{\sigma_A},\
h_B\in E_B.
\]
Similarly, there is an affine map $S_{B\to A}: K_B^{\sigma_B}\to K_A$. Here the marginal
states are not necessarily interior points of the state space, in that case we can
restrict to the generated faces of $K_A$ and $K_B$, the corresponding base normed spaces
and their duals. Below we assume $\sigma_A\in \mathrm{int}(K_A)$ and  $\sigma_B\in
\mathrm{int}(K_B)$  for simplicity.  We then have 
$(f_a\otimes id_B)(\sigma_{AB})=S_{A\to B}(f_a)$. By Lemma \ref{lemma:meas_assemb}, the
family of  measurements $\{f_a\}$ is (formally) and assemblage which admits a LHS model
 if and
only if the family is compatible. The following is now immediate.

\begin{prop} The state $\sigma_{AB}$ is unsteerable by a family of measurements $\mathcal
F$ if and only if  $S_{A\to B}$ is $\mathcal
F$-incompatibility breaking: it maps $\mathcal F$ to an assemblage with a LHS model. 

\end{prop}

Combining with Proposition \ref{prop:simple_co_dual}, we obtain the following condition
(this should be compared with the condition in \cite[Theorem 1]{nguyen2018quantum} in the
quantum case).

\begin{thm}\label{thm:unsteerable1} Let $\mathcal F$ be a set of measurements with $k$
outcomes. A  state $\sigma_{AB}$ is $\mathcal F$-unsteerable if
and only if for any boundary measure $\mu\in \mathcal P^b_{\sigma_B(K)}$, any measurement $f\in
\mathcal F$ and any $g_1,\dots, g_k\in A$, we have
\[
\sum_a \<g_a,S_{A\to B}(f_a)\>\le \int (g_1\vee\dots \vee g_k)d\mu.
\]

\end{thm}

If $\mathcal F$ is the set of dichotomic measurements, we obtain the following
characterization by properties of the map $S_{B\to A}$, or its extension to a positive map
$A_B\to V_A$. Note that the inequality is similar to the principal radius
\cite[Eq.(320)]{uola2020quantum} for qubit systems.

\begin{thm}\label{thm:unsteerable} A state $\sigma_{AB}\in K_A\otimesmax K_B$ is unsteerable by dichotomic
measurements if and only if there is a boundary measure $\mu\in \mathcal
P_{\sigma_B}^b(K_B)$ such that for all $h\in A_B$,
\[
\|S_{B\to A}(h)\|_{V_A}\le \int_K |\<h,\rho\>|d\mu(\rho).
\]
In particular, this is true if $\|S_{B\to A}(h)\|_{V_A}\le s_{\sigma_B}\|h\|^{\sigma_B}$ for
all $h$. 
\end{thm}

\begin{proof} 
Let $\{f_\pm\}$ be a dichotomic measurement (that is, $f_\pm \in E_B$ and $f_++f_-=\ou_B$)
and let $\nu_f\in \mathcal P_{2,\sigma_B}(K_B)$ be the measure corresponding to
$\{(f_\pm\otimes id)(\sigma_{AB})\}$. 
By Lemma \ref{lemma:dichotomic_co}, we see that $\sigma_{AB}$ is unsteerable
if and only if there is some $\mu\in \mathcal P^b_{\sigma_B}(K_B)$ such that for all $h\in
A_B$ and all measurements $\{f_\pm\}$,
\[
\int |\<h,\rho\>|d\nu_f\le \int  |\<h,\rho\>|d\mu.
\]
The integral on the left has the form
\[
\int |\<h,\rho\>|d\nu_f=|(f_+\otimes h)(\sigma_{AB})|+ |(f_-\otimes h)(\sigma_{AB})|=
|\<f_+,S_{B\to A}(h)\>|+|\<f_-,S_{B\to A}(h)\>|
\]
and the supremum over all dichotomic measurements is equal to  $\|S_{B\to A}(h)\|_{V_A}$.
The last statement follows from Theorem \ref{thm:sdegree}.

\end{proof}

\section{Conclusions}

We have studied steering in the setting of general probabilistic theories. For dichotomic
measurements, we proved that steering can be characterized and quantified in terms of
certain Banach space tensor cross norms, analogously to compatibility of dichotomic
measurements. In the general case, we have shown that steering can be conveniently treated 
using the classical Choquet theory for  probability measures on the compact and convex
state space.

We used this setting for some alternative characterization of LHS models. 
For dichotomic assemblages with a fixed barycenter, we found a variational  expression for 
the universal steering degree that generalizes the expressions known from quantum systems
and centrally symmetric systems. We also considered characterizations of bipartite states
that are unsteerable  and obtained conditions similar to those recently proved for
quantum systems.

Our results can be immediately applied to the study of compatible measurements, in
particular a similar formula can be found for the $g$-independent incompatibility degree
for dichotomic measurements that was only lower bounded by the 1-summing constant in
\cite{bluhm2020incompatibility}. Observe also that similar results hold also for compact convex
subsets in arbitrary (infinite dimensional) locally convex spaces and can be easily
extended beyond finite outcome measurements.

\section*{Acknowledgement}
This work was supported by the grant VEGA 1/0142/20 and  the Slovak Research and
Development Agency grant APVV-20-0069.

\end{document}